\documentclass[lettersize,journal]{IEEEtran}
\usepackage{amsmath,amsfonts}
\usepackage{array}
\usepackage{textcomp}
\usepackage{stfloats}
\usepackage{url}
\usepackage{verbatim}
\usepackage{graphicx}
\usepackage{breqn}

\usepackage{algorithm}
\usepackage{algpseudocode}
\usepackage{amsthm}
\usepackage[dvipsnames]{xcolor}
\usepackage{subcaption}
\usepackage[numbers]{natbib}
\usepackage{breqn}

\newtheorem{theorem}{Theorem}[section]
\newtheorem{lemma}[theorem]{Lemma}

\newcommand{\stol}{\Sigma_{tol}}
\newcommand{\spro}{\Sigma_{pro}}
\newcommand{\dv}{\theta}
\newcommand{\DV}{\Theta}

\def\BibTeX{{\rm B\kern-.05em{\sc i\kern-.025em b}\kern-.08em
    T\kern-.1667em\lower.7ex\hbox{E}\kern-.125emX}}
\usepackage{balance}
\begin{document}
\title{Constructing Level Sets Using Smoothed Approximate Bayesian Computation}
\author{David Edwards, Julie Bessac, Franck Cappello, Scotland Leman
\thanks{This material is based upon work support by the U.S. Department of Energy, Office of Science, Office of Advanced Scientific Computing Research (ASCR) under award number ERW6344}}

\markboth{IEEE Systems Journal ~16, May~2024}%
{How to Use the IEEEtran \LaTeX \ Templates}

\maketitle

\begin{abstract}
This paper presents a novel approach to level set estimation for any function/simulation with an arbitrary number of continuous inputs and arbitrary numbers of continuous responses. We present a method that uses existing data from computer model simulations to fit a Gaussian process surrogate and use a newly proposed Markov Chain Monte Carlo  technique, which we refer to as Smoothed Approximate Bayesian Computation to sample sets of parameters that yield a desired response, which improves on ``hard-clipped" versions of ABC. We prove that our method converges to the correct distribution (i.e. the posterior distribution of level sets, or probability contours) and give results of our method on known functions and a dam breach simulation where the relationship between input parameters and responses of interest is unknown. Two versions of S-ABC are offered based on: 1) surrogating an accurately known target model and 2) surrogating an approximate model, which leads to uncertainty in estimating the level sets. In addition, we show how our method can be extended to multiple responses with an accompanying example. As demonstrated, S-ABC is able to estimate a level set accurately without the use of a predefined grid or signed distance function.
\end{abstract}


\section{Introduction}\label{sec:intro}
The terms ``level set methods" or ``level set estimation" have different meanings depending on the context and field of study.  In the statistical and machine learning communities, it has come to mean taking noisy observations of a black-box function and using a surrogate, typically either a Gaussian Pprocess (GP) or Bayesian neural network (BNN), to estimate two sets of input parameters, one in which the response of interest is greater than or equal to a threshold and the complement of this set.  In the applied mathematical and physics context, it traditionally means finding and estimating the isocontour of a function using a grid on the input space and moving the isocontour based on either an externally generated velocity field or a velocity field defined by the isocontour.  Both methods have one thing in common: finding and representing the boundary in the input space of a function based on a defined threshold.  The algorithm described in this paper adapts a stochastic algorithm known as Approximate Bayesian Computation (ABC) within Markov Chain Monte Carlo (MCMC) in order to solve this boundary issue.  

This endeavor was motivated by the challenge of addressing the input specification quandary for an arbitrary simulation. Specifically, given a simulation characterized by a vector of input parameters $\boldsymbol{\DV}$, and a continuous response $y$, our objective was to develop a methodology capable of identifying potential input configurations that would produce a specified response value from the simulation. To demonstrate the effectiveness of our approach, we use the DSS-Wise flooding simulation model \citep{altinakar2012parallelized}, which is developed and sustained by the National Center for Computational Hydroscience and Engineering at the University of Mississippi and is operational under the auspices of the Federal Emergency Management Agency (FEMA). This model simulates the ramifications of a dam breach contingent on diverse input parameters. Our aim was to determine configurations of these parameters that would engender a predetermined outcome from the simulation. The details of this simulation are delineated in Section \ref{ss:DSS-Wise}.


\textbf{Contributions:} The principal contributions of this manuscript are as follows:
\begin{itemize}
    \item Tailoring the Approximate Bayesian Computation framework to address the challenges of input specification and level set estimation.
    \item Developing the Smoothed Approximate Bayesian Computation (S-ABC) as an enhancement to the traditional Approximate Bayesian Computation within the Metropolis-Hastings Algorithm. S-ABC refines the density space for accepted parameter configurations, facilitating a Markov Chain Monte Carlo adaptation of ABC that circumvents the pitfalls associated with the conventional hard threshold method of parameter acceptance.
    \item Validating that our algorithm achieves convergence to the target-level set distribution and demonstrating our approach through applications on both synthetic and empirical datasets.
\end{itemize}

The structure of the remainder of this document is as follows: The continuation of Section \ref{sec:intro} is dedicated to a detailed comparative analysis of the relevant literature and a comprehensive description of the DSS-Wise simulation. Section \ref{sec:background} provides the foundational concepts of level sets, the Metropolis-Hastings algorithm (M-H), and Approximate Bayesian Computation. Section \ref{sec:theory} establishes the convergence of our MCMC algorithm to the level set distribution for the elicited response. Section \ref{sec:real} delineates the application of our methodology to various Gaussian processes modeled on disparate responses of the DSS-Wise framework. Section \ref{sec:sim} discusses simulation studies of our technique on established functions to provide visualizations and analyses of the level set MCMC algorithm's performance. Finally, Section \ref{sec:conclus} summarizes the conclusions and proposes avenues for future research.

\subsection{Related Work}\label{ss:Related Work}

As previously mentioned, this work has varied applications across the domains of statistics and machine learning and is indispensable for those seeking to estimate the level set of a complicated (i.e. skewed, multi-modal, high-dimensional) function. In the following paragraphs of this subsection, we highlight work from both statistics and machine learning fields, and provide a review of such applications and techniques.\par This research is intricately linked to level set estimation, a dedicated field aimed at pinpointing domain subsets that closely mimic the actual level set, derived from a limited series of noisy function evaluations. Level set estimation typically revolves around sequential experimental design or active learning scenarios, where the practitioner is interested in identifying input sets where the function value exceeds a predetermined threshold. Refer to \cite{bryan2005active}, \cite{gotovos2013active}, \cite{mason2021nearly}, and \cite{ha2021high} for illustrative examples. Bryan et al. are renowned for devising the `straddle' heuristic, a strategic approach to select the subsequent point in a sequential experimental design framework. Here, a surrogate—commonly a Gaussian Process —is employed to approximate an elusive, noise-laden, black-box function. The primary aim is to choose input configurations (termed experimental points by Bryan et al.) that enhance the GP's precision in level set estimation. Conventionally, the surrogate GP undergoes evaluations across a spectrum of random input configurations, and the candidate input that maximizes a specific criterion is chosen for the next experimental iteration. The straddle heuristic, a brainchild of Bryan et al., ingeniously balances the exploration of the input space with the exploitation of known values proximate to the target. This heuristic is encapsulated in Equation \eqref{eq:straddle}, where $s_q$ represents the input configuration under evaluation, $\hat{f}$ denotes the estimated mean of the GP surrogate at $s_q$, and $\hat{\sigma}_{s_q}$ is the estimated standard deviation provided by the GP at that point.
\begin{equation} \label{eq:straddle}
    \text{straddle}(s_q) = 1.96\hat{\sigma}_{s_q}-|\hat{f}(s_q)-t|. 
\end{equation}

Although the straddle heuristic works well in determining which input configuration to select next out of a set of points, it does not offer a way to find that set.  The most common method is to sample points in the input space and evaluate them in the GP. Later in this paper, we show that our method allows efficient exploration of the input space based on the GP to obtain a collection of points with an estimated mean response within a tolerance of the target value.  Thus, our method allows efficient exploration of the input space by only considering points with an estimated response close to the desired target instead of randomly selecting input configurations, many of which will not give a result close to the target.  Furthermore, upon the generation of multiple input configurations via our proposed methodology, a diverse array of metrics, extending beyond the straddle heuristic, may be employed to determine the subsequent experimental point. For instance, should the investigator wish to probe the input space more thoroughly, they might opt for the point exhibiting maximal variance as per the Gaussian process model, namely, the point characterized by the greatest estimated uncertainty.

An important application of level sets is in the modeling of dynamically evolving interfaces, typically influenced by gas or fluid dynamics. The origin of this application can be traced to the foundational paper by \cite{osher1988fronts}. Osher expanded his research in this domain and co-authored a comprehensive textbook with Fedkiw (see \cite{osher2006level}), which provides a thorough introduction to the topic, underpinning the theoretical foundations, and discussing numerical strategies for implementation. In this work, Osher and Fedkiw provide techniques for identifying and delineating the level set by employing grids across a specified domain subset. Recognizing the improbability of grid points coinciding with the desired level set, they innovated methods to approximate level set points using gradient data and a mathematical framework known as the signed distance function. For a \emph{black-box}, undefined function, this requires the approximation of the gradient using finite-difference methods. While effective in low-dimensional spaces, such as two or three dimensions, the gradient's magnitude escalates in direct proportion to the number of input dimensions. Consequently, not only does the computational effort required to approximate a single level set point increase linearly with the input dimension, but the precision of the estimated level set points also suffers, potentially increasing variance and error. Our methodology facilitates the exploration of the input space without reliance on a pre-established grid, leveraging function evaluations to steer the trajectory of subsequent input configurations. The conventional grid-based approach is disadvantaged by the exponential increase in grid points as the input dimensions expand. Furthermore, our approach eschews gradient information, which is advantageous for black-box functions where gradient estimation is needed.  

\subsection{DSS-Wise Simulation}\label{ss:DSS-Wise}

As mentioned above, this research employs a dam breach simulation known as DSS-Wise \citep{altinakar2012parallelized}. This simulation allows the selection of an extensive array of dams throughout the United States, utilizing topographical data derived from satellite imagery and additional metrics to customize the simulation for a particular dam site. For example, it calculates the topography around the dam to predict the water flow resulting from a rupture. In addition, the simulation employs these measurements to determine the elevation and volumetric capacity of the water in the reservoir. Consequently, while the simulation is engineered to manage a substantial portion of the data required to simulate a dam breach, it allows the user to manipulate three specific input parameters at any given site. These parameters include the breach formation time, breach invert elevation and width, measured in hours and feet, respectively. It is evident that the permitted values for the elevation and width of the breach depend on the selected dam site. For example, the breach invert elevation is constrained by the elevations at the dam's base and crest, while the breach width is limited by the dam's structural breadth.

The output of the simulation is extensive.  Figure \ref{fig:DSS-Wise_output} illustrates a representative output map of the simulation. This simulation facilitates the delineation of the \textit{ observation lines} around the dam, which monitor the volumetric flow across these demarcations throughout the simulation process. These observation lines, on their own, allow for a substantial number of possible responses. Each line generates a time series of water flow data, termed hydrographs, with an example depicted in Figure \ref{fig:time_series}.\par The hydrographs exhibit several critical attributes that are analyzed by users, including the cumulative volume of water traversing the observation line during the flood event, the latency of water arrival at the line, the maximum flow rate across the line, and the duration to this peak flow, among others. Beyond the temporal data provided by the observation lines, the simulation also outputs raster and vector files that encapsulate data on the flood's severity and projected human impacts, differentiated by diurnal and nocturnal periods.

\begin{figure}[ht]
    \centering
    \includegraphics[width=.45\textwidth]{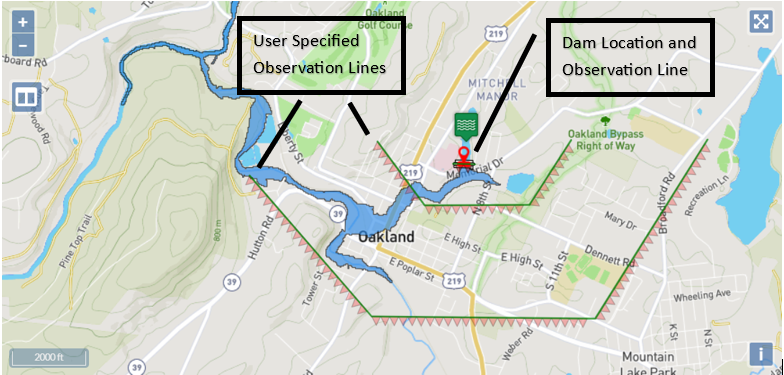}
    \caption{Example output from DSS-Wise dam breach simulation.}
    \label{fig:DSS-Wise_output}
\end{figure}

\begin{figure}[ht]
    \centering
    \includegraphics[width=.45\textwidth]{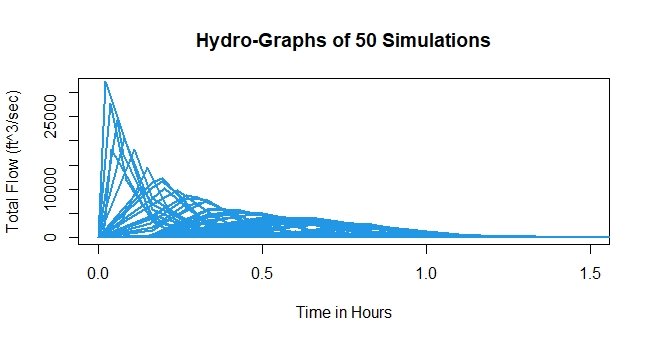}
    \caption{Time series of water flow over on observation line based on an ensemble of 50 simulation runs.}
    \label{fig:time_series}
\end{figure}

Since the goal of this research is to develop a method to address the challenges of the input specification, our main emphasis is on analyzing the responses related to hydraulic dynamics in the selected dam. Specifically, our research concentrated on the maximum total flow, quantified in cubic feet per second, the temporal interval to reach this peak flow, measured in hours, and the persistence of the water flow post-breach, also in hours.

\section{Background}\label{sec:background}

Given that solving the input specification problem is mathematically equivalent to finding the level set of a function, we frame our method in terms of level sets of a given function at a specific value.  In Section \ref{ss:level_set} we give a mathematical definition of a level set along with some examples of level sets.  In Section \ref{ss:MCMC} we give a brief overview of M-H and ABC. Our method works by placing a probability distribution over input configurations that elicit the desired response, and we developed our Smoothed-ABC, within M-H algorithm to obtain samples of these input configurations.


\subsection{Level Set} \label{ss:level_set}

In this section, we provide a brief introduction to the mathematical notion of a level set.  The familiar reader should proceed to the next subsection.  An illustration of level sets on a $2-D$ input function is given in Appendix \ref{app:level}.

Traditionally, the level set of a real-valued function, $f:\mathbf{R}^n\rightarrow\mathbf{R}$, of $n$ inputs, has been defined as the subset of the domain where $f$ takes on a constant value, $c$. Notionally we would define the level set of function $f$ at $c$ as the following: $A_c(f)=\{(\dv_1, \dv_2, ..., \dv_n)|f(\dv_1, \dv_2, ..., \dv_n) = c\}$. When the input dimension is two, the level set is often referred to as the level curve, contour line, or isoline, and, by extension, in three dimensions, the level set is often referred to as the level surface or isosurface.

A natural extension of the definition for the level set of a real-valued function is to broaden the response type and allow for vector-valued functions where the constant now becomes a vector of real values.  This only slightly changes the definition of a level set. We are now working with functions $f:\mathbf{R}^n\rightarrow\mathbf{R}^m$ and a constant vector $\Vec{c}$ of length $m$.  The definition of a level set is almost unchanged $A_{\Vec{c}}(f)=\{(\dv_1, \dv_2, ..., \dv_n)|f(\dv_1, \dv_2, ..., \dv_n) = \Vec{c}\}$. We will explore this possibility and show that our method can be extended to work with vector-valued functions in Section \ref{sec:real}.  

\subsection{Metropolis-Hastings Algorithm and Markov Chain Monte Carlo}\label{ss:MCMC}

This subsection provides a concise overview of the Metropolis-Hastings Algorithm. Readers already familiar with this algorithm may skip this subsection. It is important to note that this is not an exhaustive exploration of the Metropolis-Hastings or Markov Chain Monte Carlo  methodologies, given the extensive literature available on these topics. Rather, this section aims to provide sufficient foundational knowledge to demonstrate that our level set algorithms, delineated in Section \ref{sec:theory}, achieve convergence to the specified distribution.

Markov chain Monte Carlo is typically used to refer to a set of algorithms designed to obtain samples from a difficult to compute distribution or to estimate an intractable integral.  As the name suggests, the chain of samples is obtained using the Markovian property, namely that given a sample in the chain, the next sample only depends on the current value and not any of the other preceding values.  There are several flavors of MCMC algorithms; however, we will focus on one known as the Metropolis-Hastings algorithm \cite{metropolis1953equation}.  Much has been written about the M-H algorithm and its variants; for example, \cite{chib1995understanding} gives a partial history of the algorithm and a full derivation to show that the M-H algorithm satisfies the so-called detailed balance equation, which prove convergence to the stationary distribution (i.e. target). For our purposes, it is enough to know that the generalized M-H algorithm satisfies the detailed balance equations, and thus, as long as the support of $p(\dv)$ is a subset of the support of $q(\cdot|\cdot)$ the chain of samples generated by the M-H algorithm will converge to the target distribution of $p(\dv)$. The general M-H algorithm is presented in Algorithm \ref{alg:M-H}.

\begin{algorithm}
\caption{Metropolis Hastings Algorithm}\label{alg:M-H}
\begin{algorithmic}[1]
\Require A target distribution $p(\dv)$ and a transition density $q(\cdot|\cdot)$.
\State Allocate Matrix of Samples, $\DV$
\State Initialize $\DV[1]$ to some value in the support of $p(\dv)$
\For{$n=2, \dots, N$}
    \State Sample $\Tilde{\dv} \sim q(\cdot|\DV[n-1])$
    \State Set $\alpha =$ min$\left\{ 1, \dfrac{p(\Tilde{\dv})q(\DV[n-1]|\Tilde{\dv})}{p(\DV[n-1])q(\Tilde{\dv}|\DV[n-1])}\right\}$
    \State Sample $u\sim U(0,1)$ 
    \If{$u < \alpha$}
        \State $\DV[n] \gets \Tilde{\dv}$
    \Else
        \State $\DV[n] \gets \DV[n-1]$
    \EndIf
\EndFor 
\State \Return $\DV$
\end{algorithmic}
\end{algorithm}

The M-H algorithm is often used in Bayesian statistics, where the target distribution, $p(\dv)$, is a posterior distribution of model parameters given data, often denoted as $\pi(\dv|X)$, where $X$ represents the data.  Following Bayes' rule, these posterior distributions are proportional to a prior distribution, denoted $\pi(\dv)$, over the parameters times the model likelihood function, denoted $L(\dv|X)$. Thus, the traditional Bayesian implementation of the M-H algorithm replaces $p(\dv)$ with $\pi(\dv|X)\propto \pi(\dv)\cdot L(\dv|X)$.  Note that the normalizing constant for $\pi(\dv|X)$ is not needed for the M-H algorithm as it cancels in the ratio of line 5 of Algorithm \ref{alg:M-H}.    

\subsection{Approximate Bayesian Computation}\label{ss:ABC}

Approximate Bayesian Computation was initially introduced in 1999 by Pritchard et al. to address a complex issue in population genetics \cite{pritchard1999population}. Subsequently, ABC has evolved into an indispensable tool/algorithm for models characterized by intractable or difficult to compute likelihood functions, frequently encountered in the domain of population genetics. In these models, it is feasible to simulate the data once the parameter values have been defined. The conventional ABC methodology involves sampling parameters from the prior distribution, simulating data samples from the model, calculating a set of sufficient statistics for the simulated data, and compare these with the sufficient statistics of the observed data. Should sufficient statistics of the simulated data fall within a predetermined threshold, typically based on Euclidean distance, from those of the observed data, the parameters are retained as a sample from the posterior distribution. Conversely, if the sufficient statistics of the simulated data fall beyond this predefined threshold, the parameters are discarded. This procedure is repeated until $M$ samples from the posterior distribution are accumulated. Initially, our focus on ABC was driven by its potential to address the input specification challenge, as ABC seeks to resolve a similar problem. The ensemble of simulation runs is analogous to the observed set of data that Bayesian analysis is based; a simulation surrogate, such as a GP, assumes the role of model/data simulation in ABC, and the user-defined target aligns with the sufficient statistics. With this analogy in mind, we seek to modify ABC to effectively solve the input specification issue. 

ABC has garnered considerable attention since its initial implementation by Pritchard et al. in 1999, giving rise to multiple variants of the foundational algorithm, as discussed in the previous section. For an extensive review of ABC methodologies, refer to \cite{beaumont2019approximate}. The specific variant under consideration herein is ABC-MCMC, which integrates the ABC framework with the classical Metropolis-Hastings algorithm prevalent in Bayesian inferential statistics. This hybrid algorithm, ABC-MCMC, is described in Algorithm \ref{alg:M-H_ABC}.  Note that instead of including the ratio of likelihood evaluations, $\frac{L(\Tilde{\dv}|X)}{L(\DV[n-1]|X)}$, the algorithm uses an indicator function that is $1$ if $\|\Tilde{s}-s_{obs}\|\leq\epsilon$ and $0$ otherwise.  Thus, parameter settings that yield sufficient statistics outside of a ``sphere" centered around $s_{obs}$ are automatically rejected.  Parameter settings that elicit sufficient statistics inside of the ``sphere" are accepted with probability $\min\left\{ 1, \frac{\pi(\Tilde{\dv})q(\DV[n-1]|\Tilde{\dv})}{\pi(\DV[n-1])q(\Tilde{\dv}|\DV[n-1])}\right\}$.

\begin{algorithm}
\caption{Approximate Bayesian Computation within Metropolis-Hastings Algorithm}\label{alg:M-H_ABC}
\begin{algorithmic}[1]
\Require A prior distribution $\pi(\dv)$, a model simulator $M(\dv)$, observed data, $X$, from which sufficient statistics $s_{obs}$ are computed, and a transition density $q(\cdot|\cdot)$.
\State Allocate Matrix of Samples, $\DV$
\State Initialize $\DV[1]$ to some value in the support of $\pi(\dv)$
\For{$n=2, \dots, N$}
    \State Sample $\Tilde{\dv} \sim q(\cdot|\DV[n-1])$
    \State Simulate data, $X_{sim}$, from $M(\Tilde{\dv})$ and compute sufficient statistics $\Tilde{s}$ based on $X_{sim}$
    \State Set 
    \begin{dmath*}
    \alpha =\text{min}\left\{ 1, \dfrac{\pi(\Tilde{\dv})q(\DV[n-1]|\Tilde{\dv})}{\pi(\DV[n-1])q(\Tilde{\dv}|\DV[n-1])}\times \\ \hspace{35pt} \mathbf{I}(\|\Tilde{s}-s_{obs}\|\leq\epsilon)\right\}    
    \end{dmath*}
    \State Sample $u\sim U(0,1)$ 
    \If{$u < \alpha$}
        \State $\DV[n] \gets \Tilde{\dv}$
    \Else
        \State $\DV[n] \gets \DV[n-1]$
    \EndIf
\EndFor 
\State \Return $\DV$
\end{algorithmic}
\end{algorithm}

Having a hard cutoff has limitations in an M-H algorithm.  For example, it is possible to start the chain in a part of the input space that does not satisfy $\|\Tilde{s}-s_{obs}\|\leq\epsilon$ for $\Tilde{\dv}$ generated by $q(\cdot|\DV[1])$.  That is, the chain begins in a section of ``zero acceptance".  There are a couple of obvious solutions to this problem.  One would be to increase the variance or spread of the proposal distribution, $q(\cdot|\cdot)$.  One downside of this option is that there is no way to know how much to increase the variance to guarantee that a proposal, $\Tilde{\dv}$, could be within a region in the parameter space that will satisfy $\|\Tilde{s}-s_{obs}\|\leq\epsilon$. Furthermore, the variance of $q(\cdot|\cdot)$ has an impact on the efficiency of an M-H algorithm, and increasing the variance could have a negative impact on other aspects of the chain. The other option would be to increase $\epsilon$, thus allowing the acceptance of more possible parameter configurations.  This has the same downside as increasing the variance, and both options would require a trial-and-error approach. A third option would be to simply change the initial input configuration, but this too would require a trial-and-error approach, as we assume that if a ``good" starting location was known, the user would have started the chain there. 
 We introduce a novel methodology termed Smoothed Approximate Bayesian Computation where we substitute $\mathbf{I}(\|\Tilde{s}-s_{obs}\|\leq\epsilon)$ in line 6 of Algorithm \ref{alg:M-H_ABC} with $\frac{N(\Tilde{s}|\mu=s_{obs}, \Sigma=\stol)}{N(s^{n-1}|\mu=s_{obs}, \Sigma=\stol)}$, where $s^{n-1}$ represents the sufficient statistics for the simulated data derived from $M(\DV[n-1])$. This replacement with the ratio of multivariate normal densities, centered at $s_{obs}$, offers several advantages over the conventional indicator function. Primarily, the normal density remains positive throughout the entire spectrum of potential $\Tilde{s}$ values, ensuring that $\alpha>0$ provided $\pi(\Tilde{\dv})q(\DV[n-1]|\Tilde{\dv})>0$. Consequently, irrespective of the initial position within the chain, provided that the initial value falls within the support of $\pi(\dv)$, there exists a probability for the chain to transition towards $\Tilde{\dv}$, effectively eliminating any 'zero acceptance' regions. Secondly, by employing the ratio, the chain is facilitated in making incremental progress towards the objective $s_{obs}$ (stationary distribution; in our case, the specified level set), even if $\Tilde{s}$ does not closely approximate $s_{obs}$. Analogously, by using the ratio of normal densities, we metaphorically provide the chain a 'gradient' to ascend towards the apex represented by $s_{obs}$. It is pertinent to note that the variance matrix, $\stol$, within these normal densities serves a purpose similar to $\epsilon$ and the norm in the indicator function. We advocate for the utilization of a diagonal matrix for $\stol$, which simplifies the interpretation of the sufficient statistics generated during this process. Upon convergence of the chain, it is anticipated that 95\% of the sampled sufficient statistics, $s$, will adhere to the condition that $|s^i-s_{obs}^i|\leq 2\sqrt{\Sigma_\epsilon^{i,i}}$, where $i$ indexes each component within the vector of sufficient statistics.

The origins of the S-ABC algorithm are Bayesian in nature, but there is no reason to limit the application of this algorithm to Bayesian analysis.  As we shall see in Sections \ref{sec:real} and \ref{sec:sim}, we apply this algorithm to solve the input specification problem for large-scale simulations, as well as level set coverage and estimation for known functions.  

\section{Level set MCMC Algorithm}\label{sec:theory}

In this section, we present the level set MCMC with the S-ABC algorithm and the level set MCMC algorithm, as well as illustrate the effects that the hyperparameters of these algorithms have on the obtained samples.  

\subsection{Problem Specification}

As previously described, the M-H algorithm converges to a predefined stationary probability distribution. Typically, the ABC algorithm is used in scenarios where this distribution is either analytically intractable or is in a black-box as in computer modeling applications. The initial task is fairly direct. Given a function or simulation, $f(\dv)$, producing a continuous output, and a constant within the response range, $c$, we can define a non-empty set $A = \{\dv : f(\dv) = c\}$. In this paper, $\dv$ denotes the input configurations for the function or simulation. The existence of a non-empty set $A$ is guaranteed assuming that $c$ is within the output range of the function or simulation. We then establish a measure, $\pi_c(\dv)$, corresponding to the domain of the function or simulation, such that $\int_A \pi_c(\dv) \, d\dv = 1$. Consequently, the support of $\pi_c(\cdot)$ is precisely the set $A$. Contrary to typical MCMC algorithms, our objective is not to explore the complexities of $\pi_c(\cdot)$ by identifying modes or to accumulate sufficient samples to approximate $\int h(\dv) \pi_c(\dv) \, d\dv$ for some function $h(\dv)$. Rather, our goal is to traverse the input space and approximate the support of $\pi_c(\cdot)$, or in other terms, to acquire samples from the set $A$ and to determine the boundaries of $A$.

\subsection{Surrogate for the Objective Function}\label{ss:obj_fun}
In most cases, the function or simulation, denoted as $f(\dv)$, remains a black-box function. Typically, a space-filling design, such as a Latin hypercube, is used to produce a collection of simulation trials, upon which a surrogate model is constructed based on the gathered data. Multiple surrogate methodologies are available; however, this study utilizes a GP surrogate, represented as $G(X|\dv)$, where $\dv$ means the inputs to the function or simulation, and $X$ represents the targeted response from the simulation. Refer to \cite{gramacy2020surrogates} for an exhaustive exploration of GPs as surrogate models in computational simulations. The function or simulation, $f(\dv)$, may be deterministic or subject to additive stochastic errors. In scenarios involving the latter, the level set is defined by a region (i.e. the true level set function with added variability). With this modification to our level set MCMC algorithm, we incorporate the stochastic noise within the Gaussian process through estimation of the `nugget' parameter (i.e. an additional additive variance parameter that measures non-spatial uncertainty). Furthermore, replication within the space-filling design is used to estimate the `nugget'. This paper mostly focuses on deterministic simulations, yet includes a `nugget' parameter within the GP surrogate to maintain numerical stability, and ensuring that the estimation of the nugget is close to zero.

After fitting the Gaussian process surrogate to the observed data, we deploy the S-ABC algorithm to find the input configurations that encapsulate the level set of the function or simulation. The notion of `coverage' is cumbersome outside of a one-dimensional framework; however, as described in Section \ref{sec:sim}, the level set MCMC  S-ABC algorithm is designed to generate a distribution of points (a point cloud) that adequately spans the level set for the function/simulation, $f(\dv)$. This coverage is achieved by the level-set MCMC Algorithm with Smoothed-ABC by integrating over the uncertainty inherent in the GP, by using the estimated variance at a sampled input configuration, and subsequently sampling from the marginal normal distribution at the said configuration. This process is similar to that of conventional ABC. This step is documented in line 10 of Algorithm \ref{alg:LSMCMC-Var}. It is important to note that $G_\mu(\cdot)$ symbolizes the mean value of the GP at the designated input configuration, while $G_\Sigma(\cdot)$ denotes the variance or uncertainty.  Note the similarity between line 5 of Algorithm \ref{alg:M-H_ABC} and line 10 of Algorithm \ref{alg:LSMCMC-Var}.

\begin{algorithm}
\caption{Level set MCMC Algorithm with Smoothed ABC}\label{alg:LSMCMC-Var}
\begin{algorithmic}[1]
\Require GP Surrogate $G(X|\dv)$, and $c$ in the range of the GP.
\State Allocate Matrix of Input Configurations, $\DV$
\State Allocate Vector of Sample Responses, $\mathbf{S}$ 
\State Initialize $\DV[1]$ to some value in the support of $G(X|\dv)$
\State Sample $s\sim N(G_\mu(\DV[1]),G_\Sigma(\DV[1]))$ and Assign $\mathbf{S}[1] \gets s$
\For{$n=2, \dots, N$}
    \State Sample $\Tilde{\dv} \sim N(\DV[n-1], \Sigma_{pro})$
    \If{$\Tilde{\dv}$ is not in the support of $G(X|\dv)$}
        \State Set $\alpha=0$
    \Else
        \State Sample $\Tilde{s} \sim N(G_\mu(\Tilde{\dv}), G_\Sigma(\Tilde{\dv}))$
        \State Set $\alpha =$ min$\left\{ 1, \dfrac{N(\Tilde{s}|\mu=c, \Sigma_{tol})}{N(\mathbf{S}[n-1]|\mu=c, \Sigma_{tol})}\right\}$
    \EndIf
    \State Sample $u\sim U(0,1)$ 
    \If{$u < \alpha$}
        \State $\DV[n] \gets \Tilde{\dv}$
        \State $\mathbf{S}[n] \gets \Tilde{s}$
    \Else
        \State $\DV[n] \gets \DV[n-1]$
        \State $\mathbf{S}[n] \gets \mathbf{S}[n-1]$
    \EndIf
\EndFor 
\State \Return $\DV$ and $\mathbf{S}$
\end{algorithmic}
\end{algorithm}

In line 11 of Algorithm \ref{alg:LSMCMC-Var}, the computation involves the ratio of Gaussian densities, each centered at the designated target value and characterized by a predetermined variance. These Gaussian densities are assessed using a marginal sample from the GP in the current input configuration $\mathbf{S}[n-1]$, and in the proposed input configuration $\Tilde{s}$. Should $\Tilde{s}$ increase the target, $c$, more than $\mathbf{S}[n-1]$, this ratio will increase, thereby enhancing the likelihood of acceptance of the input configuration proposed by the algorithm. Crucially, line 7 incorporates a verification step to ascertain that the proposed input configuration lies within the permissible range of input values for the GP, as delineated by the dataset employed for the GP's calibration. Generally, GPs excel in interpolation, but perform poorly in extrapolation. Proposals that extend beyond the data set boundaries may induce a stochastic divergence in the chain, preventing convergence to the accurate distribution. Typically, a space-filling design, such as a Latin hypercube, is implemented to determine the input values for the GP calibration. The range of values selected in this design establishes the priors for each input parameter. In conventional ABC, input configurations are sampled from these priors, necessitating the a well-defined proper prior distribution. It is worth noting that the proposal distribution's symmetry precludes its inclusion in the acceptance ratio computation in line 11. Conversely, had a non-symmetric proposal been selected, a corresponding ratio of the proposal densities would be requisite in computing the acceptance ratio.

In Section \ref{sec:sim}, Algorithm \ref{alg:LSMCMC-Var} will be utilized to demonstrate the coverage property described above, employing a well-known function. In contexts where evaluating the level set is an integral step of a sequential design paradigm, a restricted subset of data points is selected in each iteration, serving as potential candidates for subsequent experimental or simulation endeavors. In such instances, the emphasis transitions from mere ``coverage" to the strategic selection of input configurations that precisely delineate the level set and rigorously scrutinize the input domain to ensure comprehensive boundary identification. Accordingly, the preference shifts towards selecting points that are most likely to facilitate the accurate demarcation of boundaries with a minimal number function/simulation evaluations. A refined version of Algorithm \ref{alg:LSMCMC-Var} is proposed, based exclusively on the mean of the GP in the proposed input configuration, deliberately eschewing the inherent uncertainty associated with the use of a GP to approximate the function or simulation. This version of the algorithm is described in the following section.

\begin{algorithm}
\caption{Level set MCMC Algorithm}\label{alg:LSMCMC}
\begin{algorithmic}[1]
\Require GP Surrogate $G(X|\dv)$, and $c$ in the range of the GP
\State Allocate Matrix of Input Configurations, $\DV$
\State Initialize $\DV[1]$ to some value in the support of $G(X|\dv)$
\For{$n=2, \dots, N$}
    \State Sample $\Tilde{\dv} \sim N(\DV[n-1], \Sigma_{pro})$
    \If{$\Tilde{\dv}$ is not in the support of $G(X|\dv)$}
        \State Set $\alpha=0$
    \Else
        \State Set $\alpha =$ min$\left\{ 1, \dfrac{N(G_\mu(\Tilde{\dv})|\mu=c, \Sigma_{tol})}{N(G_\mu(\DV[n-1])|\mu=c, \Sigma_{tol})}\right\}$
    \EndIf
    \State Sample $u\sim U(0,1)$ 
    \If{$u < \alpha$}
        \State $\DV[n] \gets \Tilde{\dv}$
    \Else
        \State $\DV[n] \gets \DV[n-1]$
    \EndIf
\EndFor 
\State \Return $\DV$
\end{algorithmic}
\end{algorithm}

The principal distinctions between Algorithms \ref{alg:LSMCMC-Var} and \ref{alg:LSMCMC} lie in the transition from sampling the estimated marginal response distribution to utilizing the estimated mean of the GP as the predictive measure for the response. This mean is assessed through the ratio of normal densities to ascertain the acceptance of a proposed input configuration. Employing the estimated mean (vs. a sample) transforms the GP into a deterministic functional/simulation emulator, suggesting that Algorithm \ref{alg:LSMCMC} could serve as a method to obtain the level sets of any deterministic function, which we demonstrate in Section \ref{sec:sim}.

Within both of the described algorithms (stochastic/deterministic), two variance matrices are posed, both of which have important influences. These include the variance matrix that controls the spread within the proposal distribution, $\Sigma_{\text{prop}}$, and the variance that sets the allowed deviation in the response compared to the target, $\Sigma_{\text{tol}}$. The matrix $\Sigma_{\text{prop}}$ governs the exploration of the input space for $\pi_c(\cdot)$ and can be adjusted to handle the complexities in support of $\pi_c(\cdot)$. For instance, if $A$ is disjoint, making $\pi_c(\cdot)$ multi-modal, a proposal capable of `jumping' between the disjoint areas of $A$ enhances the domain space exploration. This challenge is well recognized and documented in Metropolis-Hastings algorithms such as the one described in Algorithm \ref{alg:LSMCMC}. A depiction of multi-modality is shown in Subsection \ref{ss:2D_examp}. The second variance matrix, $\Sigma_{\text{tol}}$, specifies the allowable tolerance around the response $c$. 

As shown in Section \ref{sec:real}, in the case where multiple responses of interest emerge from a specific simulation, a GP is individually fitted to each response, subsequently stacking these responses as a vector of responses utilized in the preceding algorithm. Consequently, the estimated mean for the response $i^{th}$ in the input configuration $\dv$, represented by $G_{\mu,i}(\dv)$, possesses the characteristic that $95\%$ of all $\dv$ samples obtained via this algorithm will manifest an estimated mean response from the GP surrogate within the components-wise confines of $c_i\pm1.96\stol^{(i,i)}$. Incorporating covariance terms into $\stol$ could confuse the interpretation of acquired samples. Taking $\stol$ as a diagonal matrix suggests the independence of each component. The advantage of employing a multivariate Gaussian over a concatenation of univariate Gaussian distributions lies in the simplification of implementation and increase in numerical stability.

We still need to show that the samples obtained from Algorithm \ref{alg:LSMCMC}, converge to the desired level set (or distribution of level sets).  These mathematical proofs are ancillary to presenting S-ABC and the level set algorithms and are thus presented in the Appendix \ref{app:Proofs}.

\section{Application - DSS-Wise model}\label{sec:real}

In this section, we implement the level set MCMC algorithm within the DSS-Wise flood simulation framework. As described in Subsection \ref{ss:DSS-Wise}, this model incorporates three primary inputs that relate the characteristics of a dam breach specific to an individual dam. These characteristics encompass the formation time of the breach, quantified in hours, the elevation of the breach invert measured in feet, and the mean width of the breach, also in feet. The latter two parameters are constrained by the topographical features of the dam, and each dam in the simulation has distinct thresholds for acceptable parameter values. Empirical evidence suggests that the formation time for a partial breach typically does not exceed one hour, see \cite{altinakar2020failure}.  

Initially, a dam was selected within the township of Oakland, Maryland. Two observation boundaries were established that surround the dam, as illustrated in Figure \ref{fig:DSS-Wise_output}. Specifically, the Little Youghiogheny site number 1 dam was chosen for analysis, encompassing the Little Youghiogheny site number 1 reservoir. An ensemble of simulation points was generated employing a Latin hypercube sampling method with a dimensionality of 50, which serves as a space-filling strategy for the scaled input parameters. The parameters for the height of the breach and the width of the breach were normalized between 0 and 1, corresponding to the value range specified by the DSS-Wise simulation, 2400 ft-2425.94 ft for the height of the breach and 150 ft-502.8 ft for the width of the breach. 

There are many possible responses that are amenable for analysis; however, as a demonstration, we consider the maximum total flow of water over the dam in $ft^3/sec$, the time to maximum flow and the duration of the flood (both in hours).  First, we apply our level set MCMC Algorithm to the maximum total flow response and then examine how the level set MCMC Algorithm performs when applied to all three responses simultaneously by applying a multivariate Gaussian distribution over the desired response values instead of a univariate Gaussian described in line 8 of Algorithm \ref{alg:LSMCMC}.

\subsection{Maximum Total Flow}\label{ss:max_flow}

As described in Algorithm \ref{alg:LSMCMC}, a GP was fit to the Latin hypercube of size $50$, and the response of interest, the maximum flow of water over the dam. Algorithm \ref{alg:LSMCMC} was run on the resulting GP.  As a basis for better understanding the relationship between the input parameters and the response, we present below the estimated Sobol indices for each input in Table \ref{tab:sobol}.

\begin{table}
    \centering
    \begin{tabular}{|c|c|c|}
        \hline
        Formation time & Scaled breach elevation & Scaled breach width \\
        \hline
        0.6 & 0.26 & 0.01 \\
        \hline
    \end{tabular}
    \caption{Estimated Sobol indices for the three input parameters for the maximum flow response.}
    \label{tab:sobol}
\end{table}

In order to apply Algorithm \ref{alg:LSMCMC} we specify a target response, $c=5000$, a tolerance, $\stol=50^2$, and variance for the proposal distribution, $\spro=0.05^2$.  The algorithm was run for $20,000$ iterations with an initial starting configuration of formation time $0.85$, scaled breach invert elevation $0.95$, and scaled breach width $0.5$.  The results are displayed in Figure \ref{fig:MF}.

\begin{figure*}[ht]
    \begin{subfigure}{.33\textwidth}
        \includegraphics[width=\linewidth]{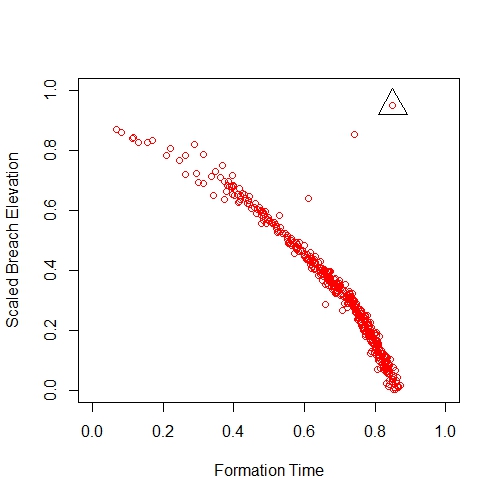}
        \caption{Input samples for formation time and scaled breach elevation based on a desired response of $5000$ maximum flow}
        \label{fig:MF_forTimeXbreachEle}
    \end{subfigure}
    \begin{subfigure}{.33\textwidth}
        \includegraphics[width=\linewidth]{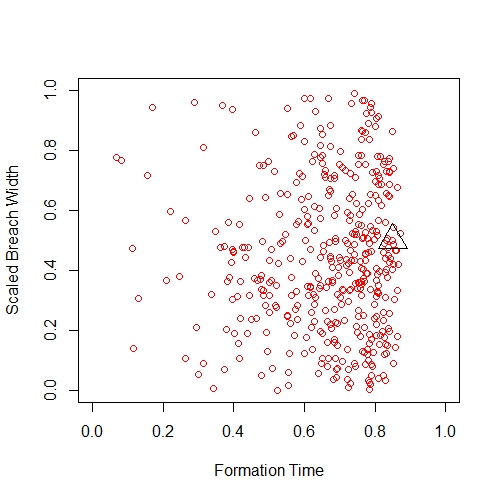}
        \caption{Input samples for formation time and scaled breach width based on a desired response of $5000$ maximum flow}
        \label{fig:MF_forTimeXbreachWid}
    \end{subfigure}
    \begin{subfigure}{.33\textwidth}
        \includegraphics[width=\linewidth]{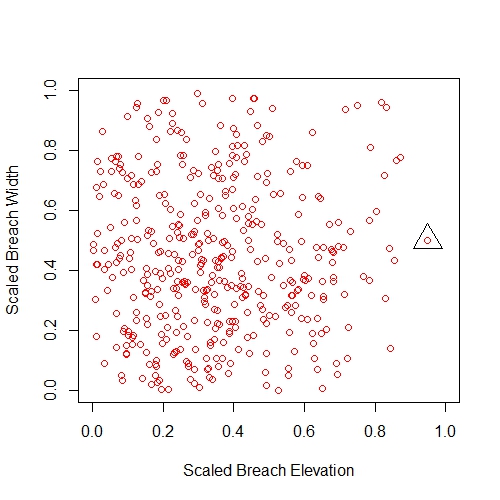}
        \caption{Input samples for scaled breach elevation and scaled breach width based on a desired response of $5000$ maximum flow}
        \label{fig:MF_breachEleXbreachWid}
    \end{subfigure}
    \caption{Input configuration samples obtained using the level set MCMCM algorithm.}
    \label{fig:MF}
\end{figure*}

These plots are consistent with the Sobol indices presented in Table \ref{tab:sobol} since formation time and scaled breach elevation clearly have a direct relationship with maximum flow and scaled breach width has little to no impact.  However, these plots offer additional information beyond what the Sobol indices provide. Specifically, there is a clear inverse relationship between formation time and scaled breach elevation as they relate to maximum flow.  

\subsection{Level Set Estimation On Three Responses}\label{ss:3responses}

Often, the user of the simulation is interested in multiple responses.  In this subsection, we illustrate how the level set MCMC algorithm can be altered to find the intersection of level sets for multiple responses simultaneously.  As mentioned earlier, the only modification needed to the level set MCMC algorithm to accommodate multiple responses is to use a multivariate Gaussian distribution centered at a vector of the desired responses with a specified diagonal variance matrix.  Each element of the diagonal represents the tolerance around the desired responses.

For this analysis, we focus on three responses: maximum flow, time to maximum flow, and duration of the flood. The goal is to find input configurations that provide the user-defined target for all three responses. As an illustration, we set the desired outcome to be a maximum flow of $5,000 \hspace{1mm}\text{ft}3/\text{sec}$, a time-to-max flow of $0.3$ hours, and the duration of the flood to be $8$ hours.  A tolerance was selected such that the marginal standard deviation of each response would be $1$\% of the desired response, thus $\stol=\left(\begin{smallmatrix}
  50^2 & 0 & 0\\
  0 & 0.003^2 & 0 \\
  0 & 0 & 0.09^2
\end{smallmatrix}\right)$.  In this way, we can apply Algorithm \ref{alg:LSMCMC} to solve this problem by changing the ratio in line $8$ to be the ratio of multivariate normal distributions, $c$ to be a vector, $\Vec{c} = (5000, 0.3, 9)$ and $\Sigma_{\text{tol}}$ to be defined as above.  Given that the inputs are scaled, we keep $\Sigma_{\text{pro}}$ as it was defined in Section \ref{ss:max_flow}.  In order to estimate each response, a GP is fit to each response based on the same Latin hypercube of size $50$ mentioned in Section \ref{ss:max_flow}.  

To help illustrate how Algorithm \ref{alg:LSMCMC} runs on multiple responses, we first run the algorithm on each response individually and then compare the sample inputs obtained with a single response with the samples obtained by using multiple responses.  Figure \ref{fig:3responses} shows the three pairwise plots obtained by comparing the input configurations for each response individually and when run collectively.  

\begin{figure*}[ht]
    \begin{subfigure}{.33\textwidth}
        \includegraphics[width=\linewidth]{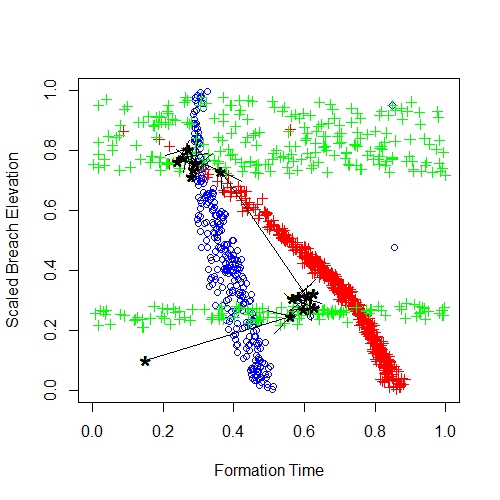}
        \caption{Input samples for formation time and scaled breach elevation based on the desired responses}
        \label{fig:3_forTimeXbreachEle}
    \end{subfigure}
    \begin{subfigure}{.33\textwidth}
        \includegraphics[width=\linewidth]{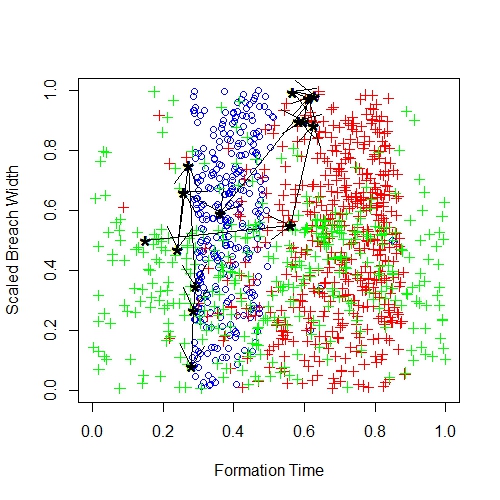}
        \caption{Input samples for formation time and scaled breach width based on the desired responses}
        \label{fig:3_forTimeXbreachWid}
    \end{subfigure}
    \begin{subfigure}{.33\textwidth}
        \includegraphics[width=\linewidth]{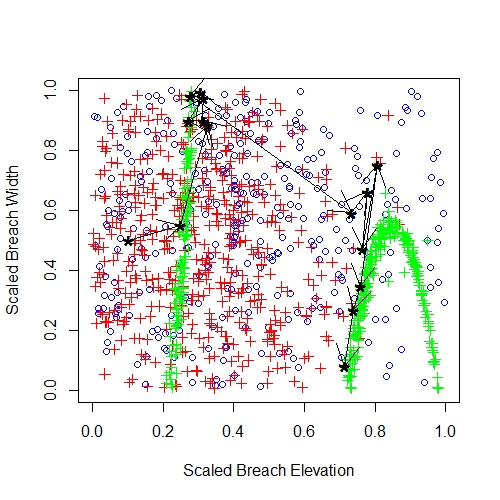}
        \caption{Input samples for scaled breach elevation and scaled breach width based on the desired responses}
        \label{fig:3_breachEleXbreachWid}
    \end{subfigure}
    \caption{Input configuration samples obtained using the level set MCMC algorithm on three different responses.  The red plus signs represent estimated input configurations in the maximum flow level set of $5,000$, the blue circles represent estimated input configurations on the time to maximum flow level set of $0.3$, the green plus signs represent estimated input configurations on the flood duration level set of $8$, and finally the black stars represent the samples obtained using algorithm \ref{alg:LSMCMC} on the three responses simultaneously.  Arrows are included to show the direction of the chain when run on all three responses.}
    \label{fig:3responses}
\end{figure*}

The red plus signs in Figures \ref{fig:3_forTimeXbreachEle}, \ref{fig:3_forTimeXbreachWid}, and \ref{fig:3_breachEleXbreachWid} represent the estimated level set for maximum flow at $5000$ and should be relatively similar to the points shown in Figure \ref{fig:MF}, as this is the same response, the fit of the GP and the target response.  The blue circles represent the estimated input configurations that give the $c=0.3$ level set for the time to maximum flow.  Finally, the green plus signs represent the estimated input configurations that give the set level $c=8$ for the response to the flood duration.  Each of these distributions of points represents the individual level sets for each response given a set target; thus it makes intuitive sense that the simultaneous level set for all three responses is the intersection of all three distributions.  This is exactly what we see happening.  The black stars represent the sampled input configurations of the simultaneous level set for all three responses.  Black arrows have been added to help illustrate the progression of the multi-response chain and to connect up points from all three plots. As we can see, especially in Figure \ref{fig:3_forTimeXbreachEle}, the chain eventually reaches the part of the input space where the red, green, and blue distributions of the points intersect.  There are fewer unique inputs sampled from the chain with three responses compared to chains with a single response. However, this is expected, as moving to a high-dimensional setting flattens out the probability space, thus reducing the number of accepted proposals.

The illustration above works well because there is overlap in the input space, making it possible to obtain all three targets simultaneously.  It is very easy to imagine a scenario in which it is not possible to find an input configuration that will simultaneously hit the target for all responses.  In order to investigate this situation, we changed the targets to a maximum flow of $12,000$, a time to maximum flow of $0.7$, and kept the same flood duration of $8$ hours.  The results of these runs are shown in Figure \ref{fig:3responses_nonoverlap}.

\begin{figure}[ht]
    \begin{subfigure}{.23\textwidth}
        \includegraphics[width=\linewidth]{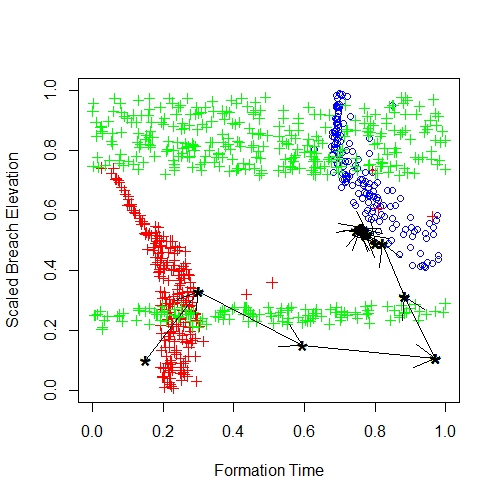}
        \caption{}
        \label{fig:3_nonoverlap_time}
    \end{subfigure}
    \begin{subfigure}{.23\textwidth}
        \includegraphics[width=\linewidth]{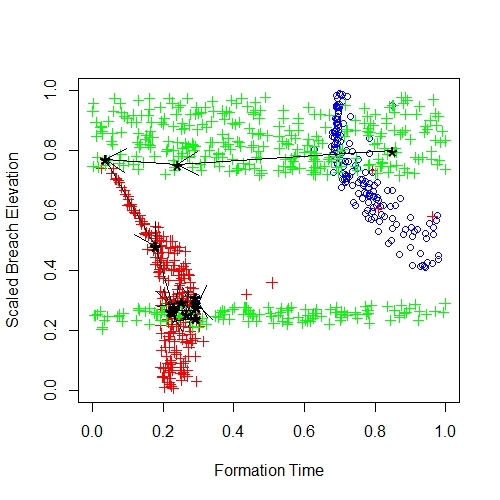}
        \caption{}
        \label{fig:3_nonoverlap_mf}
    \end{subfigure}
    \caption{Input configuration samples obtained using the level set MCMC algorithm on three different responses, where the simultaneous level set for the three responses is not possible.  The level set MCMC algorithm will obtain samples based on which response has the (relative to scale) smaller tolerance.}
    \label{fig:3responses_nonoverlap}
\end{figure}

As we can see from Figures \ref{fig:3_nonoverlap_time} and \ref{fig:3_nonoverlap_mf} the distribution for the Max Flow input configurations has shifted to the lower left corner, and the distribution for the Time to Max Flow input configurations has shifted to the upper right corner.  Thus, it is not possible to overlap all three distributions.  For Figure \ref{fig:3_nonoverlap_time} the tolerance matrix is set to be, $\stol=\left(\begin{smallmatrix}
  100^2 & 0 & 0\\
  0 & 0.002^2 & 0 \\
  0 & 0 & 0.09^2
\end{smallmatrix}\right)$, which gives a relatively lower tolerance on the response, Time to Max Flow, than the other responses.  This has the effect of the multi-response chain, the black stars, converging to points within the distribution for Time to Max Flow.  While all the points are within the blue distribution, as expected, the chain is converging to the points in that distribution that are `closest' in input space to the red distribution of points.  We can adjust the `preference' of responses in the multi-response chain by adjusting the relative tolerance on each response.  For example, the tolerance for the run shown in Figure \ref{fig:3_nonoverlap_mf} was set to be,  
$\stol=\left(\begin{smallmatrix}
  10^2 & 0 & 0\\
  0 & 0.009^2 & 0 \\
  0 & 0 & 0.09^2
\end{smallmatrix}\right)$ which is a dramatic decrease in the tolerance on Max Flow, from $100^2$ to $10^2$, and an increase on the tolerance for Time to Max Flow, from $0.002^2$ to $0.009^2$.  The result is the multi-response chain converges to points within the Max Flow distribution of red points while attempting to match the target of Time to Max Flow and Flood Duration as much as possible.  

\section{Simulated Studies}\label{sec:sim}

In this section, we demonstrate both the MCMC Algorithm of the level set that incorporates S-ABC and the same algorithm applied to a basic function, facilitating a straightforward visualization and a detailed understanding of the various adjustable parameters within the algorithm. Furthermore, we investigate the performance of the algorithm as the tolerance parameter approaches $0$, as defined in Theorem \ref{theorem}. Moreover, we shall examine the application of the level set MCMC algorithm to a high-dimensional function.

\subsection{Two Dimensional Example}\label{ss:2D_examp}

To illustrate the level set MCMC Algorithm as described in Algorithm \ref{alg:LSMCMC}, we will implement this algorithm on a well-defined function with two variables. This approach will enable us to effectively visualize both the function and the resulting samples. Now, consider the function described below:
\begin{equation}\label{eq:2D_test_fun}
    f(\dv_1,\dv_2)=e^{-((\dv_1-2)^2+(\dv_2-2)^2)}+2e^{-((\dv_1+2)^2+(\dv_2+2)^2)}.
\end{equation}

\begin{figure}
    \centering
    \includegraphics[width=.2\textwidth]{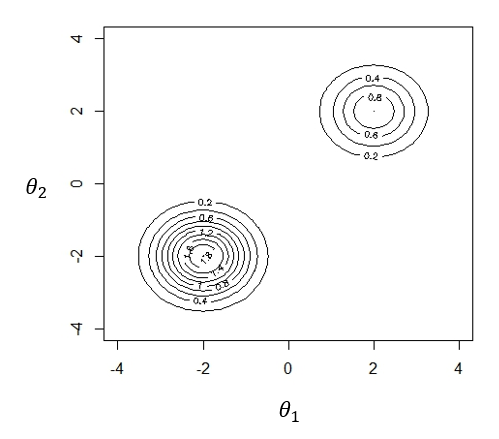}
    \caption{Contour plot for function given in Equation \eqref{eq:2D_test_fun}.}
    \label{fig:2D_test_fun}
\end{figure}

The function's contour plot is depicted in Figure \ref{fig:2D_test_fun}. Several key aspects of this function should be highlighted. Firstly, it features a global maximum at $(-2,-2)$ and a local maximum at $(2,2)$. Secondly, its range is defined by the interval $(0,2+e^{-32}]\approx(0,2]$. Thirdly, for any chosen value of $c$ within $(0,1)$, the corresponding level set comprises two circles, each centered at the global and local maxima, respectively, at $(-2,-2)$ and $(2,2)$. The contour lines in Figure \ref{fig:2D_test_fun} illustrate this behavior. Adhering to the initial method described in this section, a Latin hypercube with 50 points was created, and the function was assessed at these points. Using fifty points might seem excessive for a two-dimensional function, but it was chosen to demonstrate our algorithm clearly without the added complexity from GP uncertainty. Initially, we implemented Algorithm \ref{alg:LSMCMC-Var} targeting a value of $c=0.6$, with a proposal variance matrix of 
$\spro=\left(\begin{smallmatrix}
  3 & 0\\
  0 & 3
\end{smallmatrix}\right)$. 
A tolerance level of $\stol=0.1$ was set, with a uniform prior defined over the rectangle $\dv_1 \in [-4,4]$ and $\dv_2 \in [-4,4]$. We selected a variance of $3$ to facilitate transitions between the two modes. Subsequently, we executed Algorithm \ref{alg:LSMCMC} under the same conditions. The samples obtained are displayed in Figures \ref{fig:smoothed_abc_uniform_prior} and \ref{fig:mean_mcmc}. The initial analysis of these figures indicates that the level set MCMC, as described in Algorithm \ref{alg:LSMCMC} and illustrated in Figure \ref{fig:mean_mcmc}, effectively approximates the desired level set at $c=0.6$. The plotted points closely match the actual contour line, albeit with minor deviations due to the inherent uncertainty in the GP. A further observation reveals that the distribution of points by the level set MCMC using the Smoothed ABC, Algorithm \ref{alg:LSMCMC-Var}, compensates for the discrepancies in the mean estimate, effectively enveloping the accurate contour line within a cluster of points. Additionally, the dispersion of points around the contour at the local maximum at $(2,2)$ appears more spread out compared to those at the global maximum at $(-2,-2)$, logically deduced from the stronger signal at the global maximum, suggesting reduced uncertainty in the GP in that region compared to the local maximum.

\begin{figure}[ht]
    \begin{subfigure}[t]{.23\textwidth}
        \includegraphics[width=\linewidth]{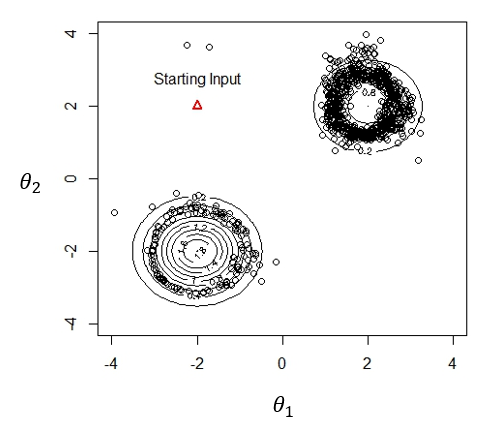}
        \caption{Smoothed ABC with uniform prior}
        \label{fig:smoothed_abc_uniform_prior}
    \end{subfigure}
    \begin{subfigure}[t]{.23\textwidth}
        \includegraphics[width=\linewidth]{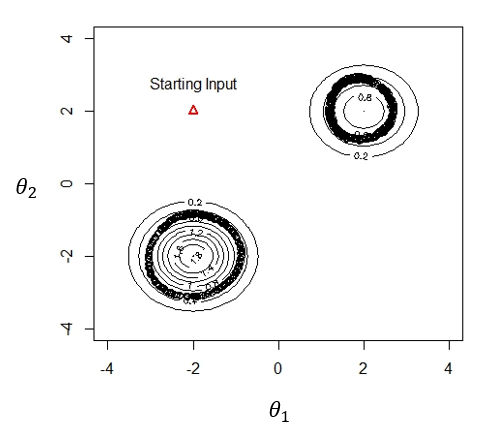}
        \caption{Level set MCMC}
        \label{fig:mean_mcmc}
    \end{subfigure}
    \caption{Effect of tolerance on accepted samples.}
    \label{fig:level_set_mean_and_smoothed}
\end{figure}

As discussed in Section \ref{sec:theory}, multiple adjustable parameters exist to tailor the function/simulation to the user's tolerance. As noted previously, the parameter $\stol$ determines the distance required to the target to accept a proposed input configuration. To explore the impact of varying $\stol$ on sample collection, we employ Algorithm \ref{alg:LSMCMC} on Function \eqref{eq:2D_test_fun} with $c=0.6$ in two different tolerance settings, $\stol=0.01^2$ and $\stol=0.1^2$. The application of the algorithm \ref{alg:LSMCMC} to the function \ref{eq:2D_test_fun} demonstrates the effectiveness of the algorithm in a deterministic function and clarifies the influence of $\stol$ on the results of the sample, eliminating any ambiguity in whether the function is deterministic or probilistic in nature. Given that the response is scalar, a univariate normal distribution is utilized as the target for the MCMC chain. These scenarios are depicted in Figure \ref{fig:tol_diff}. Both chains start at the same point, undergo the same number of MCMC iterations (5000), and use identical proposal distributions. In each case, the chain approaches the local maximum at $(2,2)$. The key difference lies in the proximity to the target response required to accept an input into the chain. In Figure \ref{fig:tol_0.01}, the chain, once converged to the $0.6$ contour line, remains closely aligned with minimal deviation. In contrast, Figure \ref{fig:tol_0.1} shows significant deviation from the contour line $0.6$ even after convergence. Post-convergence, most of the inputs in Figure \ref{fig:tol_0.1} fall between the contour lines $0.4$ and $0.8$, aligning with the expectations for a normal distribution centered at $0.6$ with a standard deviation of $0.1$, where it is expected that roughly $95\%$ of the accepted inputs will lie within $(0.6-2\times0.1, 0.6+2\times0.1)=(0.4,0.8)$. This relationship between the specified tolerance level and the acceptable response range is analagous to our method when expanded to a response vector, with $\stol$ configured as a diagonal matrix.

\begin{figure}[ht]
    \begin{subfigure}[t]{.23\textwidth}
        \includegraphics[width=\linewidth]{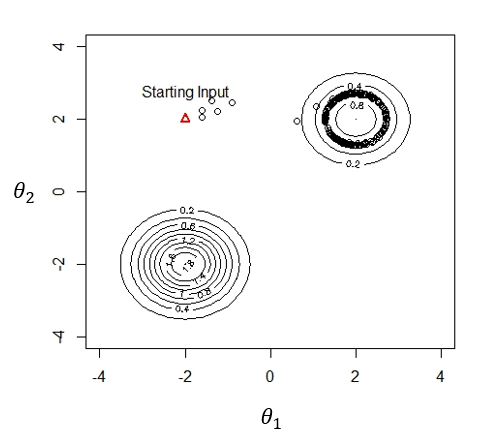}
        \caption{$\stol=0.01^2$}
        \label{fig:tol_0.01}
    \end{subfigure}
    \begin{subfigure}[t]{.23\textwidth}
        \includegraphics[width=\linewidth]{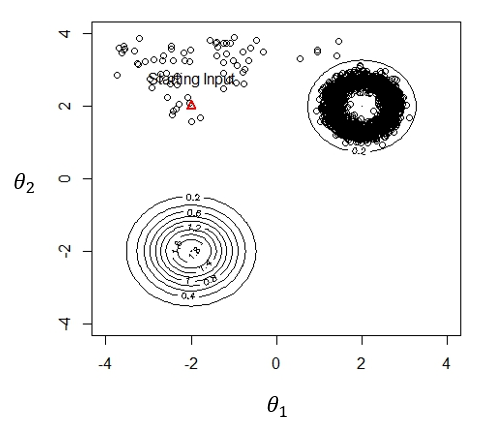}
        \caption{$\stol=0.1^2$}
        \label{fig:tol_0.1}
    \end{subfigure}
    \caption{Effect of tolerance on accepted samples.}
    \label{fig:tol_diff}
\end{figure}

Another adjustable parameter in level set MCMC algorithms, commonly shared with Metropolis-Hastings algorithms, influences how the algorithm navigates the input space. In these algorithms, $\Sigma_{\text{prop}}$ determines the magnitude of the `jump' that the chain can make with each proposal. A smaller $\Sigma_{\text{prop}}$ results in proposals that are closer to the current chain point, thereby complicating transitions between different modes of the target distribution. This behavior is illustrated in Figure \ref{fig:prop_diff}, which compares two identical chains in setup and parameters except for $\Sigma_{\text{prop}}$. Figure \ref{fig:pro_0-2} employs a normal distribution with a variance of $0.2$ for each input as its proposal distribution, while Figure \ref{fig:pro_2} uses a variance of $2$. Both scenarios allow the chain to proceed for $5000$ iterations. As observed in Figure \ref{fig:pro_0-2}, the chain is trapped in the vicinity of the global maximum at $(-2, -2)$, failing to transition to the other local maximum. In contrast, Figure \ref{fig:pro_2} demonstrates the chain's ability to traverse between the two `modes, capturing samples from both areas, although with a reduced count of unique samples. This exemplifies the classic dilemma of exploration versus exploitation inherent in Metropolis-Hastings and general MCMC algorithms. Theoretically, irrespective of the chosen value $\Sigma_{\text{prop}}$, as the number of MCMC iterations nears infinity, $N \rightarrow \infty$, the level set MCMC algorithm is expected to converge to the complete distribution and identify all 'modes'. In practice, $\Sigma_{\text{prop}}$ can influence the algorithm performance.

\begin{figure}[ht]
    \begin{subfigure}{.23\textwidth}
        \includegraphics[width=\linewidth]{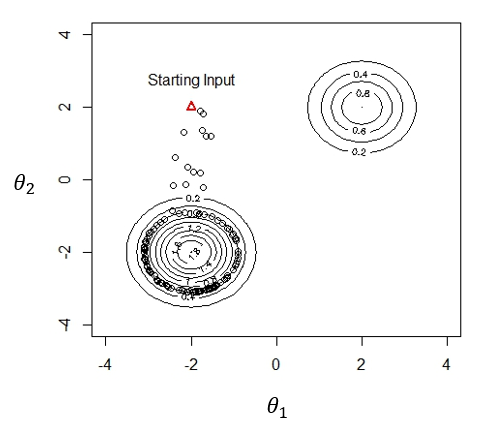}
        \caption{$\spro=\begin{bmatrix}
                            0.2 & 0\\
                            0 & 0.2
                        \end{bmatrix}$}
        \label{fig:pro_0-2}
    \end{subfigure}
    \begin{subfigure}{.23\textwidth}
        \includegraphics[width=\linewidth]{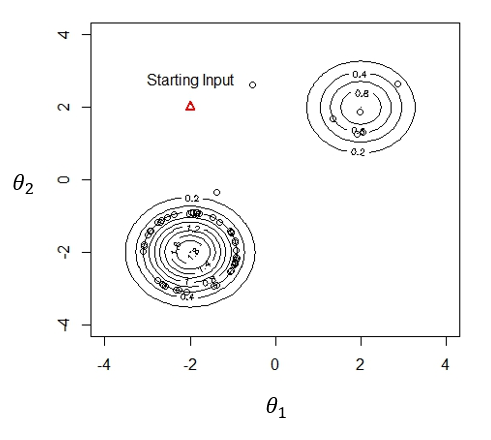}
        \caption{$\spro=\begin{bmatrix}
                            2 & 0\\
                            0 & 2
                        \end{bmatrix}$}
        \label{fig:pro_2}
    \end{subfigure}
    \caption{Comparison of input samples for different proposal distributions.}
    \label{fig:prop_diff}
\end{figure}

\subsection{Function with 100 Inputs} \label{ss:100inputs}

The level-set techniques introduced in \cite{osher2006level} target two or three-dimensional objective functions. These techniques utilize gradient data and objective function evaluations on a grid to approximate the required level sets. If the objective function lacks an analytical closed form, its gradients are approximated using finite difference methods. This approximation is effective in lower-dimensional scenarios, such as two or three dimensions, where estimating up to three partial derivatives suffices. However, in high-dimensional contexts like extensive simulations, estimating numerous partial derivatives, potentially in the hundreds, introduces substantial variability in level set predictions. To assess the Level-Set MCMC algorithm's performance in high-dimensional approximations, we designed a function with 100 variables.

The function utilizes a multivariate Gaussian kernel centered at $\mu=(1,2,\dots, 99, 100)$ to simplify visualization, where the center of the $i^{th}$ input corresponds to the integer $i$. The variance/covariance matrix is structured with $50$ on the diagonal, and specific covariances set as $Cov(x_1,x_2)=30$ and $Cov(x_5, x_10)=-40$, with all other covariances being zero. A factor of $10,000$ is multiplied by the kernel, setting the function's peak when the input matches the center $\mu$. The sole change to Algorithm \ref{alg:LSMCMC} involved altering the proposal mechanism for generating new values in the chain. Rather than a proposal affecting all inputs at each step, a sequential approach was implemented, where each input is individually proposed and then either accepted or rejected. This method of sequential proposals is frequently used in high-dimensional MCMC to increase the likelihood of achieving an acceptable region, as simultaneous modifications to all inputs generally result in lower acceptance rates.

\begin{figure}[ht]
    \begin{subfigure}{.23\textwidth}
        \includegraphics[width=\linewidth]{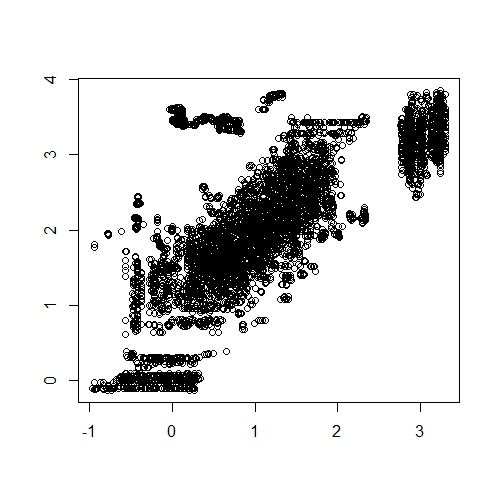}
        \caption{}
        \label{fig:high_d_1_2}
    \end{subfigure}
    \begin{subfigure}{.23\textwidth}
        \includegraphics[width=\linewidth]{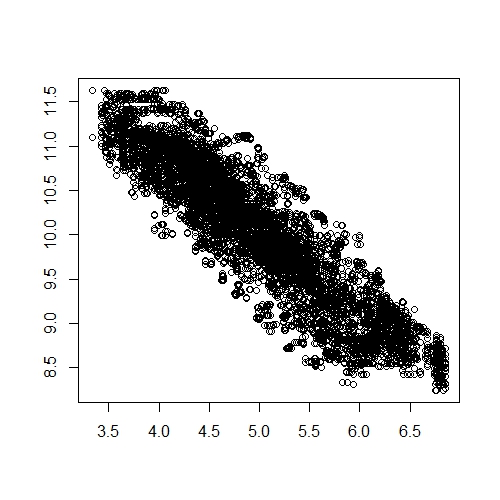}
        \caption{}
        \label{fig:high_d_5_10}
    \end{subfigure}
    \caption{Distribution of accepted input configurations for inputs 1 and 2, plot \ref{fig:high_d_1_2}, and for inputs 5 and 10, plot \ref{fig:high_d_5_10}, for the high dimensional normal function obtained using the level set MCMC Algorithm \ref{alg:LSMCMC}.}
    \label{fig:high_d_cor}
\end{figure}

\begin{figure}[ht]
    \begin{subfigure}{.23\textwidth}
        \includegraphics[width=\linewidth]{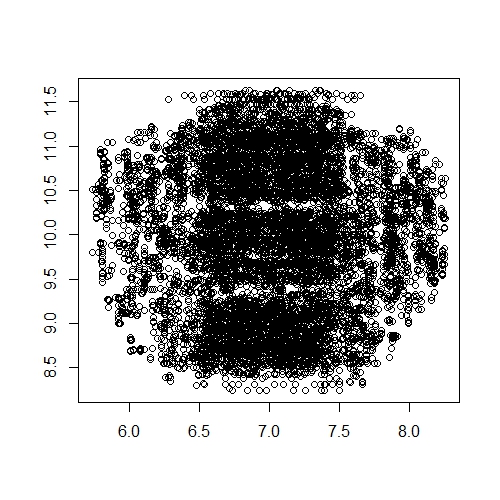}
        \caption{}
        \label{fig:high_d_7_10}
    \end{subfigure}
    \begin{subfigure}{.23\textwidth}
        \includegraphics[width=\linewidth]{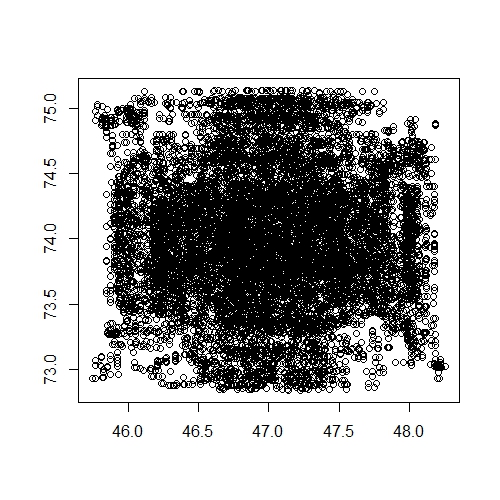}
        \caption{}
        \label{fig:high_d_47_74}
    \end{subfigure}
    \caption{Distribution of accepted input configurations for inputs 7 and 10, plot \ref{fig:high_d_7_10}, and inputs 47 and 74, plot \ref{fig:high_d_47_74}, for the high dimensional normal function obtained using the level set MCMC Algorithm \ref{alg:LSMCMC}.}
    \label{fig:high_d_independent}
\end{figure}

\begin{figure}[ht]
    \centering
    \includegraphics[width = .8\linewidth]{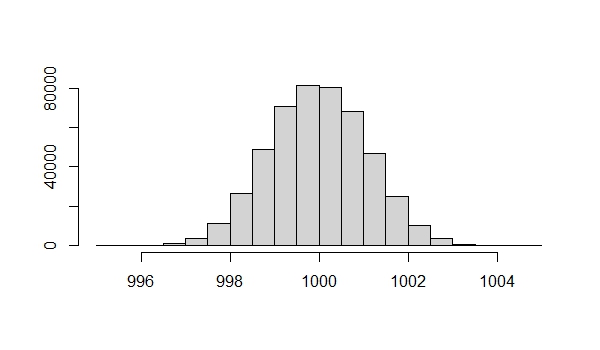}
    \caption{Histogram of the unique response values obtained using Algorithm \ref{alg:LSMCMC} with a target of $c=1,000$ and a tolerance of $\sigma_{tol}=1$.}
    \label{fig:hist_high_d_normal}
\end{figure}

To assess the convergence of Algorithm \ref{alg:LSMCMC} in high-dimensional scenarios, the algorithm was executed on the previously described normal function. A target of $c=1,000$ and a tolerance level of $\sigma_{tol}=1$ were established. The chain was initialized with a value that was randomly perturbed near the mean and a response of $9,200$ (which was arbitrarily chosen). For thorough exploration of the input space, the chain was run for $50,000$ iterations. Given that the proposal is evaluated for each input sequentially, this resulted in a total of $5,000,000$ function evaluations. Figure \ref{fig:high_d_cor} illustrates the expected positive correlation between inputs $1$ and $2$, and the negative correlation between inputs $5$ and $10$. Figure \ref{fig:high_d_independent} presents two examples of inputs with no correlation in the high-dimensional normal function. Furthermore, Figure \ref{fig:hist_high_d_normal} shows the histogram of the response values collected using Algorithm \ref{alg:LSMCMC} after a burn-in period of $1000$ iterations. The results demonstrate that the chain rapidly achieved the stationary target while simultaneously exploring the input space. This example demonstrates the effectiveness of the Level-Set MCMC Algorithm in high-dimensional contexts and is capable of generating thousands of input configurations that produce the required response from the objective function.

\section{Conclusions and Future Research}\label{sec:conclus}

As shown, both the level set MCMC algorithm with Smoothed ABC and the level set MCMC algorithm prove effective for identifying level sets for specified functions or simulation models. These algorithms were initially crafted with the challenge of input specification in mind, but they may also enable/enhance other applications. For example, the techniques described in \cite{osher2006level} employ grid points and a modification of the objective function known as a signed distance function to pinpoint the locations of the level set. However, our approaches can precisely determine a level set without relying on a fixed grid or signed distance function. As referenced in Section \ref{ss:Related Work}, evolving the level set through a velocity field is a typical use of level set methods. To advance the estimated level set via the grid approach, the numerical constraints necessitate adopting the \textit{Eulerian} formulation of the differential equation induced by the velocity field, expressed as $\phi_t+\Vec{V}\nabla\phi=0$. Here, Osher and Fedkiw denote the objective function by $\phi(x)$, with $\phi_t$ representing the time derivative of the objective function, and $\Vec{V}$ representing the velocity field. In contrast, our technique could utilize the more conceptually straightforward \textit{Lagrangian} formulation, described by $\frac{d\Vec{X}}{dt}=\Vec{V}(\Vec{x})$, where $\Vec{x}$ indicates points on the level set of $\phi(\Vec{x})$, and $\Vec{V}(\Vec{x})$ illustrates the velocity field as a function of these level set points. Osher and Fedkiw argue against the \textit{Lagrangian} formulation because even simple velocity fields can significantly distort boundary elements (like segments or triangles), potentially degrading the method's accuracy unless the discretization is periodically adjusted to mitigate these distortions through smoothing and regularizing imprecise surface elements. Nevertheless, since the level set MCMC algorithm can generate points exceedingly close to the desired level set, it is feasible to apply the velocity field to each point individually, bypassing the use of boundary elements such as segments or triangles; however, a thorough verification of this requires further research.

\bibliographystyle{IEEEtran}
\bibliography{level_set}

\begin{thebibliography}{10}
\providecommand{\url}[1]{#1}
\csname url@samestyle\endcsname
\providecommand{\newblock}{\relax}
\providecommand{\bibinfo}[2]{#2}
\providecommand{\BIBentrySTDinterwordspacing}{\spaceskip=0pt\relax}
\providecommand{\BIBentryALTinterwordstretchfactor}{4}
\providecommand{\BIBentryALTinterwordspacing}{\spaceskip=\fontdimen2\font plus
\BIBentryALTinterwordstretchfactor\fontdimen3\font minus \fontdimen4\font\relax}
\providecommand{\BIBforeignlanguage}[2]{{%
\expandafter\ifx\csname l@#1\endcsname\relax
\typeout{** WARNING: IEEEtran.bst: No hyphenation pattern has been}%
\typeout{** loaded for the language `#1'. Using the pattern for}%
\typeout{** the default language instead.}%
\else
\language=\csname l@#1\endcsname
\fi
#2}}
\providecommand{\BIBdecl}{\relax}
\BIBdecl

\bibitem{altinakar2012parallelized}
M.~S. Altinakar and M.~Z. McGrath, ``Parallelized two-dimensional dam-break flood analysis with dynamic data structures,'' in \emph{World Environmental and Water Resources Congress 2012: Crossing Boundaries}, 2012, pp. 1513--1522.

\bibitem{bryan2005active}
B.~Bryan, R.~C. Nichol, C.~R. Genovese, J.~Schneider, C.~J. Miller, and L.~Wasserman, ``Active learning for identifying function threshold boundaries,'' \emph{Advances in neural information processing systems}, vol.~18, 2005.

\bibitem{gotovos2013active}
A.~Gotovos, ``Active learning for level set estimation,'' Master's thesis, Eidgen{\"o}ssische Technische Hochschule Z{\"u}rich, Department of Computer Science, 2013.

\bibitem{mason2021nearly}
B.~Mason, R.~Camilleri, S.~Mukherjee, K.~Jamieson, R.~Nowak, and L.~Jain, ``Nearly optimal algorithms for level set estimation,'' \emph{arXiv preprint arXiv:2111.01768}, 2021.

\bibitem{ha2021high}
H.~Ha, S.~Gupta, S.~Rana, and S.~Venkatesh, ``High dimensional level set estimation with bayesian neural network,'' in \emph{Proceedings of the AAAI Conference on Artificial Intelligence}, vol.~35, no.~13, 2021, pp. 12\,095--12\,103.

\bibitem{osher1988fronts}
S.~Osher and J.~A. Sethian, ``Fronts propagating with curvature-dependent speed: Algorithms based on hamilton-jacobi formulations,'' \emph{Journal of computational physics}, vol.~79, no.~1, pp. 12--49, 1988.

\bibitem{osher2006level}
\BIBentryALTinterwordspacing
S.~Osher and R.~Fedkiw, \emph{Level Set Methods and Dynamic Implicit Surfaces}, ser. Applied Mathematical Sciences.\hskip 1em plus 0.5em minus 0.4em\relax Springer New York, 2006. [Online]. Available: \url{https://books.google.com/books?id=i4bfBwAAQBAJ}
\BIBentrySTDinterwordspacing

\bibitem{metropolis1953equation}
N.~Metropolis, A.~W. Rosenbluth, M.~N. Rosenbluth, A.~H. Teller, and E.~Teller, ``Equation of state calculations by fast computing machines,'' \emph{The Journal of Chemical Physics}, vol.~21, no.~6, pp. 1087--1092, 1953.

\bibitem{chib1995understanding}
S.~Chib and E.~Greenberg, ``Understanding the metropolis-hastings algorithm,'' \emph{The American Statistician}, vol.~49, no.~4, pp. 327--335, 1995.

\bibitem{pritchard1999population}
J.~K. Pritchard, M.~T. Seielstad, A.~Perez-Lezaun, and M.~W. Feldman, ``Population growth of human y chromosomes: a study of y chromosome microsatellites.'' \emph{Molecular biology and evolution}, vol.~16, no.~12, pp. 1791--1798, 1999.

\bibitem{beaumont2019approximate}
M.~A. Beaumont, ``Approximate bayesian computation,'' \emph{Annual review of statistics and its application}, vol.~6, pp. 379--403, 2019.

\bibitem{gramacy2020surrogates}
R.~B. Gramacy, \emph{Surrogates: Gaussian process modeling, design, and optimization for the applied sciences}.\hskip 1em plus 0.5em minus 0.4em\relax CRC press, 2020.

\bibitem{altinakar2020failure}
M.~Altinakar, M.~McGrath, V.~Ramalingam, J.~Demby, and G.~Inci, ``Two-dimensional modeling of the ka loko dam failure flood,'' \emph{USSD 2020 Annual Conference}, 2020.

\end{thebibliography}

\newpage

\appendix

\begin{center}
    {\bf Level Set Example}
  \end{center}
\section{Level Set Example}\label{app:level}

As an introduction to level sets, we briefly examine an artificial landscape known as the Goldstein-Price function given by:

\begin{align}\label{eq:gold_price}
    f(\dv_1,\dv_2)=&[1+(\dv_1+\dv_2+1)^2 \nonumber\\
    &\hspace{10pt} (19-14\dv_1+3\dv_1^2-14\dv_2+6\dv_1\dv_2+3\dv_2^2)] \nonumber \\
    &[30+(2\dv_1-3\dv_2)^2\nonumber \\
    &\hspace{10pt}(18-32\dv_1+12\dv_1^2+48\dv_2-36\dv_1\dv_2+27\dv_2^2)].
\end{align}

The Goldstein-Price function, characterized by two continuous inputs and a continuous output, attains a global minimum of $3$ at the coordinates $(0,-1)$. Figure \ref{fig:gold_price} presents a contour plot of this function. Observant readers will recognize that contour plots are essentially collections of numerous level curves represented within a two-dimensional framework. The levels specified for this illustration are set at ${c=5, 10, 50, 100, 500, 1000, 1500, 1e^4, 5e^4, 1e^5, 5e^5, 1e^6}$. The process of obtaining level sets for a bivariate function is relatively straightforward, see \cite{osher2006level}, with numerous computational packages available that facilitate the generation of plots similar to Figure \ref{fig:gold_price}. This particular plot was generated in R by evaluating the Goldstein-Price function on a grid comprising $1,681 (\dv_1,\dv_2)$ points. These computational methodologies for generating contour plots typically employ or are analogous to the techniques for determining level sets as described in \cite{osher2006level}, which utilize a grid coupled with a mathematical framework known as a signed distance function to determine and approximate level sets. However, employing the same methodology to acquire multiple level sets for a function with three continuous inputs necessitates $68,921$ evaluations of the function, presupposing uniform spacing between grid points and a comparable-sized bounding box surrounding the inputs. Inevitably, the curse of dimensionality implies that this figure will grow exponentially, as a function with four dimensions would require $2,825,761$ evaluations to achieve comparable precision in level set production. Such a voluminous number of evaluations may prove unfeasible for functions that are challenging to evaluate, such as extensive simulations. Consequently, while brute-force approaches are effective for functions with fewer inputs, they become impractical as the number of inputs expands. Therefore, to address the issue of input specification, it was imperative to devise a methodology that is effective within the expansive dimensional spaces typical of most simulations. 

\begin{figure}
    \centering
    \includegraphics[width=.47\textwidth]{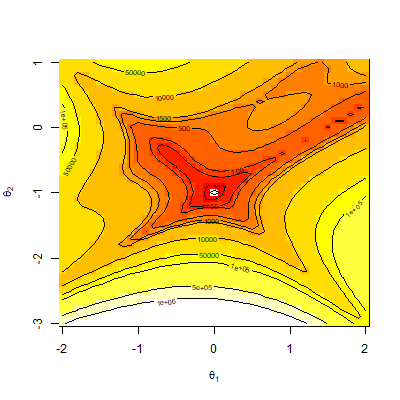}
    \caption{Level sets for the Goldstein-Price function at\\ $c=5, 10, 50, 100, 500, 1000, 1500, 1e^4, 5e^4, 1e^5, 5e^5, 1e^6$.}
    \label{fig:gold_price}
\end{figure}

\begin{center}
    {\bf Proof of Convergence in Distribution}
  \end{center}
\section{Proof of Convergence in Distribution}\label{app:Proofs}

  The following lemmas are instrumental in showing the convergence of the samples derived from Algorithm \ref{alg:LSMCMC} to the specified level set. It is useful to observe that a generalization of Algorithm \ref{alg:LSMCMC} is shown in Algorithm \ref{alg:MCMC}. In this generalization, references to the GP Surrogate $G(\cdot|\cdot)$ are omitted, and the transition distribution is incorporated into the acceptance ratio, which frames our algorithm within the conventional Metropolis-Hastings framework.

\begin{algorithm}
\caption{Generalized MCMC Algorithm}\label{alg:MCMC}
\begin{algorithmic}[1]
\State Allocate Matrix of Inputs, $\DV$
\State Initialize $\DV[1]$ to some value
\For{$n=1, \dots, N$}
    \State Sample $\Tilde{\dv} \sim N(\DV[n-1], \Sigma_{pro})$
    \State Set 
    \begin{dmath*}
       \alpha = \text{min} \left\{1, \dfrac{N(\DV[n-1]|\mu=\Tilde{X}, \spro)}{N(\Tilde{\dv}|\mu=\DV[n-1], \spro)}\times \\
       \dfrac{N(\Tilde{\dv}|\mu=c, \Sigma_{tol})}{N(\DV[n-1]|\mu=c, \Sigma_{tol})}\right\}
    \end{dmath*}
    \State Sample $u\sim U(0,1)$ 
    \If{$u < \alpha$}
        \State $\DV[n] \gets \Tilde{\dv}$
    \Else
        \State $\DV[n] \gets \DV[n-1]$
    \EndIf
\EndFor 
\State \Return $\DV$
\end{algorithmic}
\end{algorithm}

\begin{lemma}\label{lem:detail}
    The stationary distribution for Algorithm \ref{alg:MCMC} is $N(c,\stol)$.
\end{lemma}

\begin{IEEEproof}
    Algorithm \ref{alg:MCMC} is a specific instance of the Metropolis-Hastings algorithm, which is known to satisfy the detail-balanced equation given below, where $q(\cdot|\cdot)$ is the transition distribution, $\pi(\cdot)$ is the stationary distribution, and $\alpha(\Tilde{\dv}, \dv)=\frac{\pi(\dv)q(\Tilde{\dv}|\dv)}{\pi(\Tilde{\dv})q(\dv|\Tilde{\dv})}$ is the acceptance probability used to ensure detail balance.  
    \begin{equation*} 
q(\dv|\Tilde{\dv})\pi(\Tilde{\dv})\alpha(\Tilde{\dv},\dv)=q(\Tilde{\dv}|\dv)\pi(\dv)\alpha(\dv, \Tilde{\dv}).
    \end{equation*}  
    For Algorithm \ref{alg:MCMC} we have that ${q(\cdot|\cdot)=N(\cdot|\mu=\cdot, \spro)}$ and ${\pi(\cdot)=N(\cdot|\mu=c, \stol)}$ which is the desired stationary distribution.
\end{IEEEproof}

\begin{lemma}\label{lem:delta}
    The univariate Gaussian distribution, $N(\mu,\sigma^2)$, converges in distribution to the Dirac delta function, $\delta(\dv-\mu)$ as $\sigma^2\rightarrow 0$.
\end{lemma}

\begin{proof}
    This is equivalent to showing that
    \begin{equation*}
        \lim_{\sigma\rightarrow0^+}\dfrac{1}{\sigma\sqrt{2\pi}}e^{\frac{-(\dv-\mu)^2}{2\sigma^2}}=\delta(\dv-\mu).
    \end{equation*}
    The property that defines $\delta(\dv-\mu)$ is ${\int_{-\infty}^{\infty}\delta(\dv-\mu)g(\dv)d\dv=g(\mu)}$ for all continuous compactly supported functions $g(\dv)$. Let $g(\dv)$ be a continuous compactly supported function, then we have the following:
    \begin{align*}
        \int_{-\infty}^{\infty}&\lim_{\sigma\rightarrow0^+}\dfrac{1}{\sigma\sqrt{2\pi}}e^{\frac{-(\dv-\mu)^2}{2\sigma^2}}g(\dv) d\dv \\
        &= \int_{-\infty}^{\infty}\lim_{\sigma\rightarrow0^+}\dfrac{1}{\sigma\sqrt{2\pi}}e^{\frac{-y^2}{2}}g(\sigma y + \mu) dy, && y=\frac{\dv-\mu}{\sigma} \\
        &= g(\mu)\int_{-\infty}^{\infty}\lim_{\sigma\rightarrow0^+}\dfrac{1}{\sigma\sqrt{2\pi}}e^{\frac{-y^2}{2}} dy \\
        &= g(\mu)\lim_{\sigma\rightarrow0^+}\int_{-\infty}^{\infty}\dfrac{1}{\sigma\sqrt{2\pi}}e^{\frac{-y^2}{2}} dy \\
        &= g(\mu)\cdot 1 =g(\mu).
    \end{align*}
    Note that we can justify moving $\lim$ outside of the integral through the Dominated Convergence Theorem. Thus, $\lim_{\sigma\rightarrow0^+}\dfrac{1}{\sigma\sqrt{2\pi}}e^{\frac{-(\dv-\mu)^2}{2\sigma^2}}$ converges to the Dirac delta function, $\delta(\dv-\mu)$, in distribution, and the lemma is proven.
\end{proof}

\begin{theorem}\label{theorem}
    As all diagonal elements of $\stol$ approach $0$ the stationary distribution for Algorithm \ref{alg:LSMCMC} becomes the level set distribution $\pi_{\Vec{c}}(\cdot)$.
\end{theorem}

\begin{proof}
    Given that we are assuming $\stol$ is a diagonal matrix, we find that the multivariate normal distribution is equivalent to the product of independent univariate normal distributions.  By Lemma \ref{lem:delta}, these univariate normal distributions converge to the product of the Dirac delta distributions.  Each Dirac delta function has a support over the input space that is equivalent to the level set for the individual response in $\Vec{c}$.  The support for the product of Dirac delta functions is the intersection of each of these support sets or level sets.  Thus, the stationary distribution for Algorithm \ref{alg:LSMCMC} as the diagonal elements of the $\stol$ approach $0$ is the simultaneous level set distribution for the vector of response values $\Vec{c}$ or $\pi_{\Vec{c}}(\cdot)$.
\end{proof}

\end{document}